\documentclass[11pt]{article}
\usepackage[utf8]{inputenc}
\usepackage[margin=1in]{geometry}
\usepackage{graphicx,amsfonts,amsmath,amssymb,epsfig,color,multicol,pstricks}
\usepackage{amsthm}
\usepackage[linesnumbered,noline,ruled]{algorithm2e}
\usepackage{authblk}
\usepackage{tikz}

\newcommand{\Enote}[1]{{\color{blue}\footnote{\color{blue} \textbf{Eric:} #1}}}

\renewcommand{\paragraph}[1]{{\protect\vspace{8pt}\noindent\sc{#1}}}
\usepackage{multicol}
\usepackage[nottoc]{tocbibind}
\usepackage[pagebackref=true,colorlinks]{hyperref}
\usepackage{dsfont}
\usepackage{amsmath}
\newlength{\saveparindent}
\newlength{\saveparskip}
\newcommand{\BE}{\begin{enumerate}} \newcommand{\EE}{\end{enumerate}}
\newcommand{\BI}{\begin{itemize}} \newcommand{\EI}{\end{itemize}}
\newcommand{\BDes}{\begin{description}}\newcommand{\EDes}{\end{description}}

\newtheorem{alg}{Algorithm}
\newcommand{\BA}{\begin{alg}} \newcommand{\EA}{\end{alg}}
 \newtheorem{thm}{Theorem}[section]            
\newcommand{\BT}{\begin{thm}} \newcommand{\ET}{\end{thm}}

\newtheorem{lem}[thm]{Lemma} 
\newcommand{\BL}{\begin{lem}} \newcommand{\EL}{\end{lem}}

\newtheorem{clm}[thm]{Claim}
\newcommand{\BCM}{\begin{clm}} \newcommand{\ECM}{\end{clm}}

\newtheorem{techcor}[thm]{Corollary}
\newcommand{\BCo}{\begin{techcor}} \newcommand{\ECo}{\end{techcor}}

\newtheorem{Conc}[thm]{Conclusion}
\newcommand{\BCONC}{\begin{Conc}} \newcommand{\ECONC}{\end{Conc}}

\newtheorem{Obs}[thm]{Observation}
\newcommand{\BOBS}{\begin{Obs}} \newcommand{\EOBS}{\end{Obs}}

\newtheorem{Exmp}[thm]{Example}
\newcommand{\BEXM}{\begin{Exmp}} \newcommand{\EXMP}{\end{Exmp}}

\newtheorem{fact}[thm]{Fact}

\newtheorem{cor} [thm] {Corollary}      
\newcommand{\BC}{\begin{cor}} \newcommand{\EC}{\end{cor}}
\newtheorem{prop}[thm]{Proposition}     
\newcommand{\BP}{\begin{prop}} \newcommand {\EP}{\end{prop}}
\newtheorem{conj} {Conjecture}      
\newcommand{\BCJ}{\begin{conj}} \newcommand{\ECJ}{\end{conj}}

\theoremstyle{definition}
\newtheorem{defn}{Definition}[section]         
\newcommand{\BD}{\begin{defn}} \newcommand{\ED}{\end{defn}}

\def\FullBox{\hbox{\vrule width 8pt height 8pt depth 0pt}}
\newcommand{\QED}{\;\;\;\FullBox}
\newenvironment{Proof}{\noindent{\bf Proof:~~}}{\hfill\QED}
\newcommand{\BPF}{\begin{Proof}} \newcommand {\EPF}{\end{Proof}}
\newenvironment{proofof}[1]{\noindent{\bf Proof of {#1}:~~}}{\(\QED\)}
\newcommand{\BPFOF}{\begin{proofof}} \newcommand {\EPFOF}{\end{proofof}}

\newenvironment{smallproof}{\noindent{\bf Proof sketch:~~}}{\(\QED\)}
\newcommand{\bpf}{\begin{smallproof}} \newcommand{\epf}{\end{smallproof}}

\newcommand{\BEQ}{\begin{equation}} \newcommand{\EEQ}{\end{equation}}
\newcommand{\BEQN}{\begin{eqnarray}}\newcommand{\EEQN}{\end{eqnarray}}

\renewcommand{\Pr}{{\rm Pr}}

\newcommand{\R}{\mathbb{R}}

\newcommand*{\rom}[1]{\expandafter\@slowromancap\romannumeral #1@}

\newcommand{\D}{{\rm D}}

\newcommand{\floor}[1]{{\lfloor{#1}\rfloor}}
\newcommand{\tf}{\tilde f}

\usepackage{belov_main}

\title{Testing convexity of functions over finite domains}

\author[1]{Aleksandrs Belovs}
\author[2]{Eric Blais}
\author[2]{Abhinav Bommireddi}

\affil[1]{Faculty of Computing\\
          University of Latvia\\
          Riga, Latvia\\
          \texttt{aleksandrs.belovs@lu.lv}}
\affil[2]{Cheriton School of Computer Science\\
          University of Waterloo\\
          Waterloo, Canada\\
          \texttt{{eblais,vabommir}@uwaterloo.ca}}

\date{}

\begin{document}

\maketitle

\begin{abstract}
We establish new upper and lower bounds on the number of queries required to test convexity of functions over various discrete domains.

\begin{enumerate}
    \item We provide a simplified version of the non-adaptive convexity tester on the line.  We re-prove the upper bound $O\sA[\frac{\log(\eps n)}{\epsilon}]$ in the usual uniform model, and prove an $O\sA[\frac{\log n}{\eps}]$ upper bound in the distribution-free setting.

    \item We show a tight lower bound of $\Omega\sA[\frac{\log(\eps n)}{\epsilon}]$ queries for testing convexity of functions $f: [n] \rightarrow \mathbb{R}$ on the line. This lower bound applies to both adaptive and non-adaptive algorithms, and matches 
    the upper bound from item 1,
    showing that adaptivity does not help in this setting.

    \item Moving to higher dimensions, we consider the case of a stripe $[3] \times [n]$.
    We construct an \emph{adaptive} tester for convexity of functions  $f\colon [3] \times [n] \to \R$ with query complexity $O(\log^2 n)$. 
    We also show that any \emph{non-adaptive} tester must use $\Omega(\sqrt{n})$ queries in this setting.  
    Thus, adaptivity yields an exponential improvement for this problem.
    
    \item 
For functions $f\colon [n]^d \to \R$ over domains of dimension $d \geq 2$, we show a non-adaptive query lower bound $\Omega\s[(\frac{n}{d})^{\frac{d}{2}}]$.
\end{enumerate}
\end{abstract}

\section{Introduction}


Let $X$ be a subset of $\R^d$.  
A function $f \colon X \to \R$ is called \emph{convex} if for every finite collection of points $x_1,x_2,\ldots, x_k \in X$ and non-negative reals $\lambda_1,\ldots,\lambda_k \ge 0$ satisfying $\sum_{i} \lambda_i = 1$ and $ \sum_i \lambda_i x_i \in X$, we have
\[
f\sB[\sum_i \lambda_i x_i] \le \sum_i \lambda_i f(x_i).
\]

Convex functions are typically considered on convex domains, but for property testing questions, we will be mostly interested in the case when $X$ is a finite (hence, discrete) subset of $\R^d$.
In this case, one can show that $f$ is convex on $X$ if and only if it can be extended to a convex function $\tilde f\colon \R^d \to \R$ on the entire linear space $\R^d$.\footnote{I.e., $f$ is convex on $X$ if and only if there exists a function $\tilde f$ that is convex on $\R^d$ and satisfies $\tilde f(x) = f(x)$ for every $x \in X$. See Section~\ref{sec:convexity} for details.}

For a finite set $X$, we say that a function $g\colon X\to\bR$ is \emph{$\epsilon$-far from convex} with respect to some \emph{proximity parameter} $0<\epsilon<1$ if for every convex function $h\colon X \to \R$, we have $\absA|\{ x \in X: h(x) \neq g(x)] \}| \ge \epsilon |X|$.

In this work, we consider the problem of distinguishing convex functions from those that are far from convex in the property testing framework~\cite{GoldreichGR98, RubinfeldS96}. Formally, an \emph{$(\epsilon, X)$-convexity tester} is a bounded-error randomized algorithm that queries the values of an unknown function $f\colon X \to \R$ on a set of inputs from $X$ and distinguishes the case where $f$ is convex from the one where $f$ is $\epsilon$-far from convex. 
A tester is \emph{non-adaptive} if it selects all the inputs to query before observing the value of $f$ on any of those inputs; otherwise the tester is \emph{adaptive}.

Our goal is to determine the minimum query complexity of $(\epsilon,X)$-convexity testers for various discrete sets $X$, and to determine whether the query complexity of adaptive and non-adaptive $(\epsilon,X)$-convexity testers differs for any set $X$.
While there has been work studying the problem of testing convexity of functions in various settings~\cite{ParnasRR03,Ben-Eliezer19,PallavoorRV18,BermanRY14,BlaisRY14}, large gaps remain between the best upper and lower bounds.  
We give new bounds on the number of queries required to test convexity of functions on the line, over a stripe, and over higher-dimensional domains.


\subsection{Testing convexity on the line}

The problem of testing convexity of functions $f\colon [n] \to \R$ on the line was first considered by Parnas, Ron, and Rubinfeld~\cite{ParnasRR03}. They showed that $O(\frac{\log n}{\epsilon})$ queries suffice to $\epsilon$-test convexity in this setting.
A slightly better upper bound of $O(\frac{\log(\epsilon n)}{\epsilon})$ was shown by Ben-Eliezer~\cite{Ben-Eliezer19}.
This follows from his more general algorithm for testing local properties of arrays.
We give a more direct algorithm.

\begin{thm}
\label{thm:1dimUpper}
There exists an $\eps$-tester for convexity of functions $f\colon[n]\to\bR$ over the line with complexity $O(\frac{\log(\epsilon n)}{\epsilon})$.
The tester is non-adaptive and has 1-sided error.
\end{thm}

We also consider the problem of testing convexity of functions on the line in the distribution-free model of Halevy and Kushilevitz~\cite{HalevyK07}. In this model, the distance of a function $g\colon X \to \R$ to convexity is measured with respect to some unknown distribution $\cD$ over the domain $X$. 
The algorithm can query the target function $f\colon [n] \to \R$ as usual, and it can also sample from $\cD$. The tester must distinguish the case where $f$ is convex and the case where $f$ is \emph{$\eps$-far from a convex function with respect to $\cD$}, in that $\Pr_{x\sim\cD} [g(x)\ne h(x)]\ge\eps$ for every convex function $h$.
The tester must work for any distribution $\cD$, and the complexity measure is the worst-case sum of the number of queries to $f$ and samples from $\cD$.
Thus, distribution-free property testing is at least as hard as usual property testing, and for some problems the query complexity is much larger in the distribution-free setting~\cite{HalevyK07}.

We show that our algorithm for testing convexity of functions $f\colon [n] \to \R$ can be made distribution-free with only a slight loss in the dependence of $\eps$.

\begin{thm}
\label{thm:1dimDistfree}
There exists a non-adaptive 1-sided algorithm that $\eps$-tests a function $f\colon [n]\to\bR$ for convexity with respect to an unknown distribution $\cD$ using
$O\sA[\frac{\log n}{\eps}]$ queries to $f$ and $O\sA[\frac{1}{\eps}]$ samples from $\cD$.
\end{thm}

The algorithms that establish Theorems~\ref{thm:1dimUpper} and~\ref{thm:1dimDistfree} are both \emph{triple testers}: 
they repeatedly draw triples of points from a natural probability distribution over $[n]^3$ and test that the function is convex on those three points.\footnote{This is similar to the situation for the well-known \emph{pair testers} for monotonicity~\cite{ErgunKKRV00} that sample pairs of points from a natural distribution and test them for monotonicity. Note also that while it is not presented as such, the convexity tester of Parnas, Ron, and Rubinfeld~\cite{ParnasRR03} can also be reformulated as a triple tester.}
This has a number of consequences.  
First, both our algorithm admit time-efficient implementation.
Second consequence is for quantum testers (see~\cite{MontanaroW16} for introduction to quantum property testing).
Using quantum amplitude amplification~\cite{brassard:amplification}, we can achieve quadratic improvement. 
Thus, in the standard property testing model, the quantum query complexity for $\epsilon$-testing convexity is $O\s[\sqrt{\epsilon^{-1}\log(\epsilon n)}]$, and in the distribution-free setting the quantum query complexity of the problem is $O\s[\sqrt{\epsilon^{-1}\log n}]$.
Again, both of the algorithms can be implemented time-efficiently.



\bigskip
Blais, Raskhodnikova, and Yaroslavtsev~\cite{BlaisRY14} showed that the bound in Theorem~\ref{thm:1dimUpper} on the query complexity of \emph{non-adaptive} convexity testers is optimal when $\epsilon > 0$ is a constant.
For adaptive algorithms, this only gives a lower bound of $\Omega(\log \log n)$ by the standard conversion between adaptive and non-adaptive algorithms.
We close this gap and show that the bound in Theorem~\ref{thm:1dimUpper} is optimal for all values of $\eps \le \frac19$, even when the testers are allowed to be adaptive.

\begin{thm}
\label{thm:1D}
For every $\frac1n \le \epsilon \le \frac19$, any $\epsilon$-tester for convexity of functions $f\colon [n] \to \R$ has query complexity $\Omega\sA[\frac{\log(\eps n)}{\epsilon}]$. 
\end{thm}


In particular, the lower bound in Theorem~\ref{thm:1D} implies that adaptivity does \emph{not} help to reduce query complexity when testing convexity of functions over the line. 
This is analogous to the situation for testing monotonicity of functions over the line \cite{Fischer04}. This result, combined with the distance approximation algorithm of Fattal and Ron~\cite{FattalR10}, also shows that approximating the distance to convexity is essentially no harder than testing convexity.



\subsection{Testing convexity over 2-dimensional domains}

Parnas, Ron, and Rubinfeld~\cite{ParnasRR03} asked whether convexity can be tested efficiently for functions over 2-dimensional domain. The first non-trivial upper bound on the query complexity for testing the convexity of functions mapping $[n]^2$ to $\R$ was obtained by Ben-Eliezer~\cite{Ben-Eliezer19}, who showed that $O(n)$ queries suffice for non-adaptive testing of convexity---a number of queries that is sublinear (in fact, quadratically smaller) than the size of the domain.

The only previous lower bound for non-adaptive testing of convexity of functions $f \colon [n]^2 \to \R$ was again $\Omega(\log n)$~\cite{BlaisRY14}, so it remained open whether it is possible to test convexity non-adaptively using a number of queries that is \emph{exponentially} smaller than the size of the domain. We show that it is not, and that the Ben-Eliezer bound is optimal for all non-adaptive algorithms when $\epsilon$ is a constant.

\begin{thm}[Special case of Theorem~\ref{Thm:HighDimLower} below]
\label{thm:2Dlb}
Any non-adaptive $\Omega(1)$-tester for convexity of $f\colon [n]^2 \rightarrow \R$ has query complexity $\Omega(n)$. 
\end{thm}

Note that Theorem~\ref{thm:2Dlb} does not eliminate the possibility that convexity of functions on $[n]^2$ can be tested with $\polylog(n)$ queries by \emph{adaptive} algorithms. Based on the results for testing convexity in 1D, one may be tempted to guess that adaptivity does not help in this setting either and that the bound in the theorem could be strengthened to apply to adaptive algorithms as well. To test this intuition, we consider an intermediate domain between 1-dimensional and full 2-dimensional case: the \emph{stripe} $[3]\times[n]$.
The same intuition from the 1-dimensional case would suggest that adaptivity does not help in testing convexity of functions over the stripe. We show, however, that here adaptivity can be used to obtain an \emph{exponential} improvement on the query complexity of convexity testers.

\begin{thm}
\label{thm: [3]x[n]}
There exists a 1-sided-error algorithm that $\eps$-tests a function $f\colon[3]\times[n]\to\bR$ for convexity in the distribution-free testing model using $O\sA[\frac{\log^2 n}{\eps}]$ queries to $f$ and $O\sA[\frac{1}{\eps}]$ samples from $\cD$. 
By contrast, any non-adaptive $\Omega(1)$-tester for convexity of $f\colon[3]\times[n]\to\bR$ (in the standard testing model) has query complexity $\Omega(\sqrt{n})$. 
\end{thm}

The exponential gap between the adaptive and non-adaptive query complexity of convexity testing in Theorem~\ref{thm: [3]x[n]} stands in stark contrast to the situation for the related problem of testing monotonicity: there it is known that adaptivity does not yield \emph{any} reduction in query complexity, as there is a non-adaptive monotonicity tester for functions $f \colon [n]^d \to \R$ with query complexity $O(d \log n)$~\cite{ChakrabartyS13a} and every monotonicity tester (adaptive or not) has query complexity $\Omega(d \log n)$~\cite{ChakrabartyS14}.

\subsection{Testing convexity over high-dimensional domains}

Ben-Eliezer's upper bound for testing convexity~\cite{Ben-Eliezer19} also carries over to high-dimensional settings. When the dimension $d$ is large, however, the bound is quite weak: it shows that $O(d n^{d-1})$ queries suffice to test convexity non-adaptively. This is (barely) sublinear in the domain size $n^d$ when $d = o(n)$.

Blais, Raskhodnikova, and Yaroslavtsev~\cite{BlaisRY14} previously showed that non-adaptive algorithms that test \emph{linear convexity} of functions over the hypergrid $[n]^d$ have query complexity $\Omega(d \log n)$. (Linear convexity is a slightly different notion of convexity than the one studied here; see Appendix~\ref{app:otherConvexity} for details.)
We show that a much stronger lower bound holds for the problem of testing convexity: any non-adaptive algorithm for testing convexity of functions over $[n]^d$ has query complexity that is linear in $n$ and exponential in $d$.

\begin{thm}
\label{Thm:HighDimLower}
For every $d \ge 2$ and any $\epsilon \le \frac1{10}$, any bounded-error non-adaptive $\epsilon$-tester for convexity has query complexity $\Omega\s[(\frac{n}{d})^{\frac{d}{2}}]$.
\end{thm}

Note that the trivial upper bound for testing convexity (or any other property) of functions over $[n]^d$ is $n^d$, so Theorem~\ref{Thm:HighDimLower} shows that non-adaptive convexity testers cannot do significantly better (qualitatively) than the na\"ive brute-force testing algorithm.

This result also implies a general lower bound of $\Omega(d \log n)$ queries for adaptive convexity testers of convexity for functions over the hypergrid $[n]^d$. This is the first general lower bound for convexity testing which shows that the query complexity must scale as the product of the dimension and the logarithm of the length of hypergrids.

\subsection{Discussion and open problems}
\label{sec:discussion}

Our results suggest two main open problems.

\newtheorem{open}{Open Problem}
\begin{open}
\label{open:lowD}
Is it possible to $\Omega(1)$-test convexity of functions $f \colon [n] \times [n] \to \R$ with $\polylog(n)$ queries?
\end{open}

Parnas, Ron, and Rubinfeld~\cite{ParnasRR03} also raised the problem of determining the query complexity for testing convexity in $d \ge 2$, and the upper bound in Theorem~\ref{thm: [3]x[n]} provides the first suggestion that the query complexity of the problem might be exponentially smaller than---and not just sublinear in---the domain size. As the lower bound in the same theorem shows, however, any algorithm that would provide a positive answer to this question would have to be adaptive.

We can also generalize Open Problem~\ref{open:lowD} to ask whether convexity testing of $f \colon [n]^d \to \R$ can be done with query complexity $\polylog(n)$ for every constant value of $d$. For high-dimensional settings, it is also natural to ask about the dependence on $d$.

\begin{open}
\label{open:highD}
Must every $\Omega(1)$-tester for convexity of functions $f \colon [n]^d \to \R$ have query complexity $2^{\Omega(d)}$?
\end{open}

Theorem~\ref{Thm:HighDimLower} gives a positive answer to this question for non-adaptive algorithms, but it still allows for the possibility that there is a convexity tester with query complexity that is polynomial in $d$. It is also possible that the best query complexity of convexity testers is subexponential in $d$, even if it is not polynomial in $d$. (C.f., for instance, the submodularity testing problem, where it is known that $2^{\tilde{O}(\sqrt{d})}$ queries suffice to test submodularity of functions $f : \{0,1\}^d \to \R$~\cite{SeshadhriV14}. It is possible that a similar bound holds for testing convexity as well.)

\subsection{Organization}

We introduce some basic facts about convexity in Section~\ref{sec:convexity}, estalish our algorithmic results in Sections~\ref{sec:1dimUpper} and~\ref{sec: [3]x[n] Upper}, and give the proofs for our hardness results in Sections~\ref{sec: High Dim lower bound} and~\ref{sec: 1-D lower bound}.

Specifically, the proofs of Theorems~\ref{thm:1dimUpper} and~\ref{thm:1dimDistfree} for testing convexity over one-dimensional domains are presented Section~\ref{sec:1dimUpper}. The upper bound in Theorem~\ref{thm: [3]x[n]} for testing convexity of functions on the stripe is established in Section~\ref{sec: [3]x[n] Upper}.

The lower bound in Theorem~\ref{Thm:HighDimLower} for testing convexity over high-dimensional domains is presented in Section~\ref{sec: High Dim lower bound}; the lower bound for the stripe in Theorem~\ref{thm: [3]x[n]} is found in Section~\ref{sec: [3]x[n] lower bound}; and the optimal lower bound for testing convexity on the line in
Theorem~\ref{thm:1D} is presented in Section~\ref{sec: 1-D lower bound}.




%



\section{Basic facts about convexity}
\label{sec:convexity}


In this section, we establish some basic facts about convex functions over finite subsets of $\bR^d$.
We use the notation $[n]=\{0,1,\dots,n-1\}$ and $[a..b]=\{a,a+1,\dots,b-1\}$.
All the results in this section are standard; we provide the missing proofs in Appendix~\ref{app:missingProofs} for completeness.

The \emph{restriction} of a function $f\colon X \to \bR$ to a domain $Y \subseteq X$ is the function $f|_Y\colon Y \to \bR$ defined by $f|_Y(y) = f(y)$ for each $y \in Y$. Our first basic observation is that restriction preserves convexity.

\begin{lem}
\label{lem:restriction}
Let $f\colon X\to\bR$ be a convex function and $Y\subseteq X$.  Then the function $f|_Y\colon Y\to\bR$ restricted to $Y$ is also convex.
\end{lem}

To define the extension of convex functions, we first need the notion of a centred simplex.

\begin{defn}
A \emph{simplex} in $\bR^d$ is a set of affinely independent points.
A \emph{centred simplex} in $\bR^d$ is a collection of points $x_1,\dots, x_k, z$ such that $x_1,\dots,x_k$ form a simplex, and $z$ can be (uniquely) expressed as 
\begin{equation}
\label{eqn:centre}
z = \sum_i \lambda_i x_i ,
\end{equation}
where all $\lambda_i > 0$ and $\sum_i \lambda_i = 1$.
The point $z$ is called the \emph{centre} of the simplex, and 
we say that the simplex is \emph{centred at $z$} when this condition is satisfied.
\end{defn}

In other words, $x_1,\ldots,x_k,z$ is a centred simplex if $z$ is inside the convex hull of $x_1,\dots,x_k$ and no $x_i$ can be removed from the simplex without breaking this property.
When $X$ is a finite subset of $\bR^d$ and $x_1,\dots,x_k,z \in X$, we say that the centred simplex is \emph{of $X$}.

\begin{defn}
The centred simplex $x_1,\dots,x_k, z$ of $X$ is \emph{minimal} iff $z$ is the only point of $X$ inside the convex hull of the simplex $x_1,\dots,x_k$ except for its vertices.
\end{defn}

\begin{lem}
\label{lem:extension}
Let $f\colon X\to\bR$ be a convex functions with $X$ a finite subset of $\bR^d$.
Then the function can be extended to a convex function on the whole space $\bR^d$.
That is, there exists a convex function $g\colon \bR^d\to\bR$ such that $g(x) = f(x)$ for all $x\in X$.
Moreover, for a point $z$ in the convex hull of $X$ the function $g$ can be defined as
\begin{equation}
\label{eqn:g}
g(z) = \min_{x_1,\dots,x_k} \sum_i \lambda_i f(x_i),
\end{equation}
where $x_1,\dots,x_k$ range over all simplices of $X$ centred at $z$, and $\lambda_i$ are as in Equation~(\ref{eqn:centre}).
\end{lem}

Combining the above two lemmata, we see that if $f\colon X\to\bR$ is a convex function with $X\subseteq\bR^d$ finite, and $X\subseteq Y\subseteq \bR^d$, then the function $f$ can be extended to a convex function on $Y$.  This is how we will usually use the above lemma.

We say that a function $f\colon X\to\bR$ is convex on a centred simplex $x_1,\dots, x_k, z$ if its restriction to this set of points is convex.  
This is equivalent to
\[
f(z) \le \sum_i \lambda_i f(x_i),
\]
where $\lambda_i$ are as in Equation~(\ref{eqn:centre}).
This notion provides a characterization of convexity that we will use to test convex functions.

\begin{thm}
\label{thm:minimalSimplex}
A function $f\colon X\to\bR$ is convex if and only if it is convex on every minimal centred simplex of $X$.
\end{thm}


Let us apply the general Theorem~\ref{thm:minimalSimplex} to the setting where $f$ is a function over the line.
For the rest of this section, let $X = \{x_1,x_2,x_3,\ldots\} \subseteq \R$ where $x_1 < x_2 < x_3 < \cdots$.
A centred simplex in this case is a \emph{triple} $x<y<z$ and $y$ is the centre of the triple.
A function $f$ is convex on the triple if and only if
\begin{equation}
\label{eqn:convexity}
\frac{f(y)-f(x)}{y-x} \le \frac{f(z)-f(y)}{z-y}.
\end{equation}
A minimal centred simplex is a \emph{minimal triple} of the form $x_i<x_{i+1}<x_{i+2}$.
Thus, we get the following corollary.

\begin{cor}
\label{cor:lineTriple}
The function $f\colon X\to \bR$ is convex if and only if it is convex on every triple $x_i<x_{i+1}<x_{i+2}$ of consecutive points.
\end{cor}

A nice feature of convex functions on the line is that we can efficiently find their minimum.

\begin{thm}
\label{thm:minimum}
Assume $f\colon X\to\bR$ is a convex function.  It is possible to find the minimum of $f$ on $X$ in time $O(\log |X|)$.
\end{thm}

\pfstart
Use bisection.  Let $n = |X|$.
If $n<6$, query all the values of $f$ and find the minimum.
Otherwise, let $a=\floor[n/2]$ and $b=a+1$.
Query $f(x_a)$ and $f(x_b)$.
If $f(x_a)<f(x_b)$, execute the minimum search on the set $x_1,\dots,x_a$.
Otherwise, execute the minimum search on the set $x_b,\dots,x_n$.
By each execution, the size of the set decreases roughly by a factor of 2, hence $O(\log |X|)$ iterations suffice.
\pfend

\section{Algorithms for testing convexity over the line}
\label{sec:1dimUpper}

In this section, we prove Theorem~\ref{thm:1dimUpper} and Theorem~\ref{thm:1dimDistfree}.
Both theorems are established using similar ideas, by constructing explicit convexity testing algorithms that are inspired by the monotonicity tester on the line~\cite{Belovs18}.

\begin{defn}
\label{defn:tripleTest}
Let $a\in[n]$.  A \emph{triple test} rooted at $a$ is a (non-necessarily sorted) triple $(a,b,c)$ such that
\begin{itemize}
\item $b\in \sfig{ 2^k\floor[\frac{a-1}{2^k}], 2^k\ceil[\frac{a+1}{2^k}] }$ for some integer $k$ satisfying $1\le 2^k<n$, and
\item $c$ is either $a+1$ or $b+1$. 
\end{itemize}
The element $a$ is called the \emph{root}, and $b$ is called a \emph{hub} of $a$.  The integer $2^k$ is called the \emph{height} of the triple.
We say that $a$ passes the triple test if the function $f$ is convex on $\{a,b,c\}$.
We say that $a$ passes all its triple tests if it passes all the triple tests rooted at it.
\end{defn}

\begin{clm}
If $x<y-1$, then $x$ and $y$ have a common hub with height not exceeding $2(y -x)$.
\end{clm}

\begin{proof}
Let $k$ be such that $\frac{y-x}{2} \leq 2^k \leq y-x$. There can be either one or two multiples of $2^k$ between $x$ and $y$. If there is just one then we are done and that is the common hub. If there are two then there will be exactly one multiple of $2^{k+1}$ between $x$ and $y$ and that is their common hub.
\end{proof}

\begin{lem}
\label{lem:tripleTest}
Assume $x<y<z$ is a non-convex triple.  Then at least one of $x$, $y$ or $z$ fails some of its triple tests with height not exceeding $2\cdot\max\{y-x,z-y\}$.
\end{lem}

\pfstart
Assume for now that $y-x\ge 2$ and $z-y\ge 2$.
Let $h$ be the common hub between $x, y$ with height not exceeding $2(y-x)$ and $h'$ be the common hub between $y, z$ with height not exceeding $2(z-y)$. 
Consider the function $f$ restricted to the domain $x<h<y<h'<z$.
By Lemma~\ref{lem:restriction}, we know the function is not convex, hence, by Corollary~\ref{cor:lineTriple}, it is non-convex on at least one of the triples $(x, h, y)$, $(h, y, h')$, or $(y, h', z)$.
Let us consider the three cases separately:
\begin{itemize}
\item $f$ is non-convex on the triple $x,h,y$.
Consider the function $f$ on the domain $x<h<h+1\le y$.
Using Corollary~\ref{cor:lineTriple} if needed, we get that the function $f$ is non-convex on one of the triples $(x,h,h+1)$ or $(h,h+1,y)$.
Each of them constitutes a triple test: $a=x$, $b=h$, $c=h+1$, or $a=y$, $b=h$, $c=h+1$, respectively.
\item $f$ is non-convex on the triple $h,y,h'$.
Consider the function $f$ on the domain $h<y<y+1\le h'$.
The function $f$ is non-convex on one of the triples $(h,y,y+1)$ or $(y,y+1,h')$.
Again, each of them constitutes a triple test: $a=y$, $b=h$, $c=y+1$, or $a=y$, $b=h'$, $c=y+1$, respectively.
\item $f$ is non-convex on the triple $y,h',z$.  This case is analogous to the first one.
\end{itemize}
If $y=x+1$, then the above analysis works with $h=x$ (the first case never holds, and $x=y-1$ is a hub of $y$).
If $z=y+1$, the above analysis works with $h'=z$ (the third case never holds, and $z=y+1$ is a hub of $y$).
Finally, if both $y=x+1$ and $z=y+1$, we can use the triple test with $a=x$, $b=y$, and $c=z$.
\pfend

A simple consequence of this lemma is that the function $f$ is convex on the set of points passing all their triple tests.
This allows us to formulate the following notion.

\begin{defn}
\label{defn:convexReplacement}
A \emph{convex replacement} of a function $f\colon [n]\to\bR$ is a convex function $\tilde f\colon [n]\to\bR$ such that $f(x) = \tilde f(x)$ for all $x$ that pass all their triple tests.
\end{defn}


The proof of Theorem~\ref{thm:1dimDistfree} now follows easily.

\pfstart[Proof of Theorem~\ref{thm:1dimDistfree}]
The algorithm is simple: sample $a$ from $\cD$ and run all the triple tests rooted at $a$.  It takes $1$ sample and $O(\log n)$ queries.  
The probability that this test fails is at least the distance (with respect to $\cD$) to the convex replacement to $f$.
Repeat the above test $O(1/\eps)$ times to increase the success probability to $\Omega(1)$.
\pfend

We are now also ready to complete the proof of Theorem~\ref{thm:1dimUpper}.

\pfstart[Proof of Theorem~\ref{thm:1dimUpper}]
The algorithm is a triple tester.
It selects a root $a$ of the triple uniformly at random from $[n]$, select a triple rooted at $a$ with height at most $2\eps n$ uniformly at random, and tests it for convexity.
For completeness, let us restate the algorithm:

\medskip
\parbox{.9\textwidth}{
\begin{algorithm}[H]
\caption{\textsc{ConvexityTest1D}}
Draw $a$ uniformly at random from $[n]$\;

Draw $k$ uniformly at random from $\{0,\dots, \ceil[\log_2(2\eps n)]\}$\;

Let $b$ be either the largest multiple of $2^k$ strictly smaller than $x$, or the smallest multiple of $2^k$ strictly larger than $x$, each case with probability $1/2$\;

Let $c$ be either $a+1$ or $b+1$, each with probability $1/2$\;
Query $f(a)$, $f(b)$ and $f(c)$ and test whether $a,b,c$ form a convex triple\;
\end{algorithm}
}
\medskip


We claim that if the function $f$ is $\eps$-far from convex, then this test fails with probability $\Omega(\eps/\log(\eps n))$.
Thus, this test has to be repeated $O(\log(\eps n)/\eps)$ times.

We will construct a subset $A\subseteq [n]$ of size $\eps n$ such that every $a\in A$ fails one of its triple tests with height at most $2\eps n$.
Start with $S\gets [n]$.
We treat $S$ as a sorted list.
While $|[n]\setminus S|< \eps n$, the function $f|_S$ is non-convex.
Choose three neighbouring elements $x<y<z$ in $S$ that violate convexity.  Let $(a,b,c)$ be the non-convex triple constructed in Lemma~\ref{lem:tripleTest}.  The height of this triple is at most $2\eps n$.  Remove $a$ from $S$.
When $|[n]\setminus S|\ge\eps n$, let $A\gets [n]\setminus S$.
\pfend

\section{Algorithm for testing convexity on the \texorpdfstring{$[3] \times [n]$}{[3]x[n]} stripe}
\label{sec: [3]x[n] Upper}

In this section, we prove the upper bound in Theorem~\ref{thm: [3]x[n]}.

\subsection{High-level description}

Our approach to testing convexity on the stripe $[3]\times [n]$ is as follows.
This set is very close to the 1-dimensional line, so we can draw a lot from the tester of Section~\ref{sec:1dimUpper}.
In this vein, for $i\in[3]$, let $f_i\colon [n]\to\bR$ be the restrictions of $f$ to the column $\{i\}\times [n]$.
We will construct a convex replacement $\tf$ of $f$ so that every point where $f$ and $\tilde f$ disagree fails some test.
Sampling $a\in [3]\times [n]$ from $\cD$ and executing the test on $a$ will give us a distribution-free tester of convexity.

Any simplex centred at a point in the line $\{0\}\times [n]$ or $\{2\}\times [n]$ is completely contained inside this line.
Hence, for $f_0$ and $f_2$ we can simply take convex replacements $\tf_0$ and $\tf_2$ from Definition~\ref{defn:convexReplacement}, and assume that $\tf$ restricted to $\{0\}\times [n]$ or $\{2\}\times [n]$ is $\tf_0$ or $\tf_2$, respectively.

Let us define a function $h\colon \{0,1/2,1,3/2,\dots,n-1\}\to\bR$ as 
\begin{equation}
\label{eqn:h}
h(x) = \min_\delta \frac{\tf_0(x-\delta)+\tf_2(x+\delta)}2.
\end{equation}
Note that $h(x) = g(1,x)$ where $g$ is the convex extension, as in Equation~(\ref{eqn:g}), of the function $\tf$ restricted to $\{0,2\}\times [n]$.  (We have not defined $\tf$ on the line $\{1\}\times [n]$ yet.)
By Lemma~\ref{lem:extension} and Lemma~\ref{lem:restriction}, the function $h$ is convex.
Its value can be computed by minimising the convex function $\delta \mapsto \sA[\tf_0(x-\delta)+\tf_2(x+\delta)]/2$.
This is exactly the place where our tester uses adaptivity.

The main part of our algorithm deals with interplay between the functions $h$ and $f_1$.
Let us give some relations between $h$ and $f_1$ for the case when $f$ is convex.
First, the function $f$ is convex on any simplex of the form $(0,x-\delta), (2,x+\delta)$ centred at $(1,x)$, which implies that $f_1(x) \le h(x)$ for every $x\in[n]$.
Next, for every $\{x,x+1\}\subseteq [n]$, let $\beta\colon\bR\to\bR$ be the affine function agreeing with $f_1$ at $x$ and $x+1$.
We have that the function $f$ is convex on any simplex of the form $(0,z-\delta), (1,x), (2,z+\delta)$ centred at $(1,x+1)$, which implies that $\beta(z)\le h(z)$ for all $z > x+1$.\footnote{Note that this observation does not immediately follow from the first observation and  convexity of $f_1$, because it also incorporates half-integer values of $z$, where $f_1$ is not defined.}
Similarly, considering simplices $(0,z-\delta), (1,x+1), (2,z+\delta)$ centred at $(1,x)$, we get that $\beta(z)\le h(z)$ for all $z<x$.
Our tester will check these conditions.
\medskip

\subsection{Subroutines}

We are now ready to describe the subroutines used by our tester.
The first subroutine is the convexity test for the line from Section~\ref{sec:1dimUpper}.

\medskip
\parbox{.9\textwidth}{
\begin{algorithm}[H]
\caption{\textsc{1DTest}$(i, x)$}
Execute all the triple tests rooted on $x$ for the function $f_i$ as in Definition~\ref{defn:tripleTest}\;

If at least one of the tests fails, output that $f$ is not convex and terminate the algorithm\;
\end{algorithm}
}
\medskip

The complexity of this subroutine is $O(\log n)$.
The following claim is a direct consequence of Definition~\ref{defn:convexReplacement}.
\begin{clm}
If \textsc{1DTest}$(i, x)$ does not fail for $i=0$ or $i=2$, then $\tf_i(x) = f_i(x)$.
\end{clm}

The next subroutine evaluates the function $h$.

\medskip
\parbox{.9\textwidth}{
\begin{algorithm}[H]
\caption{\textsc{Evaluate}$(x)$}
Find a point $\delta^*$ where the function $g(\delta)=\frac12(f_0(x-\delta)+f_2(x+\delta))$ attains its minimum, assuming this function is convex\;

For $y\in\{x-\delta^*-1,x-\delta^*,x-\delta^*+1\}$ perform \textsc{1DTest}(0, $y$)\;

For $y\in\{x+\delta^*-1,x+\delta^*,x+\delta^*+1\}$ perform \textsc{1DTest}(2, $y$)\;

Check that $g(\delta^*)\le g(\delta^*-1)$ and $g(\delta^*)\le g(\delta^*+1)$\;

Return $g(\delta^*)$\;
\end{algorithm}
}
\medskip

\begin{clm}
\label{clm:convexMinimisation}
The subroutine either finds a violation of convexity or returns $h(x)$.
The complexity of the subroutine is $O(\log n)$.
\end{clm}

\pfstart
Define $\tilde g(\delta) = \frac12(\tilde f_0(x-\delta)+\tilde f_2(x+\delta))$ so that $h(x) = \min_\delta \tilde g(\delta)$.
Steps 2 and 3 of the subroutine ensure that $g$ and $\tilde g$ agree on $\delta^*-1, \delta^*$ and $\delta^*+1$.
If Step 4 fails, we get that $g\ne \tilde g$, meaning that the function $f$ is not convex.
Otherwise, we get that $\tilde g(\delta^*)\le \tilde g(\delta^*-1)$ and $\tilde g(\delta^*)\le \tilde g(\delta^*+1)$.
As $\tilde g$ is convex, this implies that the minimum of $\tilde g$ is attained at $\delta^*$. 
The complexity estimate is obvious.
\pfend

\subsection{The algorithm}

Now let us state the test for convexity over the stripe.

\medskip
\parbox{.9\textwidth}{
\begin{algorithm}[H]
\caption{\textsc{ConvexityTestStripe}}
Sample $(i,x)$ from $\cD$\;

Perform \textsc{1DTest}($i$, $x$)\;

\If{$i=1$}{
\textsc{Evaluate} $h(x)$ and verify that $f_1(x)\le h(x)$\;

Perform \textsc{1DTest}(1, $x-1$) and \textsc{1DTest}(1, $x+1$)\;

Let $\beta^-\colon [n] \to \R$ be the affine function satisfying $\beta^-(x-1)=f_1(x-1)$ and $\beta^-(x)=f_1(x)$\;

Minimise the convex function $h-\beta^-$ on the interval $\{x+1,x+\frac32,x+2,\dots, n-1\}$ and verify that the minimum is non-negative\;

Let $\beta^+ \colon [n] \to \R$ be the affine function satisfying $\beta^+(x)=f_1(x)$ and $\beta^+(x+1)=f_1(x+1)$\;

Minimise the convex function $h-\beta^+$ on the interval $\{0,\frac12,1,\dots, x-1\}$ and verify that the minimum is non-negative\;
}
\end{algorithm}
}
\medskip

The tester uses one sample from $\cD$ and $O(\log^2 n)$ queries to $f$, since steps 7 and 9 each require $O(\log n)$ calls to the \textsc{Evaluate} subroutine, which in turn makes $O(\log n)$ queries to $f$.

By the discussion at the beginning of the section, any convex function $f$ passes the test with probability 1. 
Let $f \colon [3]\times[n] \to \R$ be any function that is $\eps$-far from convex with respect to $\cD$.
Let $S$ be the set of points that pass the test.
We claim that $f$ restricted to $S$ is convex.
Hence, the error probability of the test is at least $\eps$, and it suffices to repeat the test $O(1/\eps)$ times.

In order to prove that $f$ is convex on $S$, we extend it to a slightly larger domain.
This is done to better handle possible minimal centred simplices.
Let as above $\tf_i$ be convex replacement of $f_i$.  
We claim that the function $\tf\colon (\{0,2\}\times [n])\cup S \to \bR$ defined by
\[
\tf(i,x) =
\begin{cases}
\tf_i(x),&\text{if $i=0$ or $i=2$;}\\
f(i,x),&\text{if $(i,x)\in S$;}
\end{cases}
\]
is convex (the two values are equal when both conditions apply).
As $f$ and $\tf$ agree on $S$, this implies that $f$ restricted to $S$ is convex.

By Theorem~\ref{thm:minimalSimplex} and above discussion, it suffices to consider minimal simplices centred at points of the form $(1,x)\in S$.
From Lemma~\ref{lem:tripleTest}, we get that the function $\tf$ is convex on a centred simplex of the form $\{(1,x), (1,y), (1,z)\}\subseteq S$.
The function $\tf$ is also convex on a simplex $(0,x-\delta), (2,x+\delta)$ centred at $(1,x)$ because $h(x)\ge f_1(x)$ by Step 4.

Any other minimal simplex centred at $(1,x)$ is of the form $(0,a), (1,b), (2,c)$.
Let $y = (a+c)/2$, and assume $b<x<y$ (the case $y<x<b$ is similar).
From Step 7 of the algorithm, we know that the function $g\colon\{x-1,x,y \}\to \bR$ defined by
\begin{align*}
g(x-1)&= f_1(x-1),\\
g(x)&= f_1(x),\\
g(y)&= h(y)
\end{align*}
is convex.
Both $(1,b)$ and $(1,x)$ are in $S$ and so they pass the test. Thus, from Steps 2 and 5, we have that $f_1$ and $\tf_1$ agree on $b$, $x-1$ and $x$.
As $\tf_1$ is convex, and using Corollary~\ref{cor:lineTriple}, we have that the function $g' \colon \{b,x-1,x,y \} \to \bR$ defined by
\begin{align*}
g'(b)&= f_1(b),\\
g'(x-1)&= f_1(x-1),\\
g'(x)&= f_1(x),\\
g'(y)&= h(y)
\end{align*}
is also convex.
Finally, since $h(y) \le (\tf_0(a)+\tf_2(c))/2$, we get that $\tf$ is convex on the simplex $(0,a),(1,b),(2,c)$ centred at $(1,x)$.


\section{Lower bounds for testing convexity in high dimensions}

\label{sec: High Dim lower bound}

In this section we prove the $\Omega\s[(\frac{n}{d})^\frac{d}{2}]$ lower bound for non-adaptive algorithms that test convexity on the $[n]^d$ grid in Theorem~\ref{Thm:HighDimLower} and the $\Omega(\sqrt{n})$ lower bound for non-adaptive algorithms that test convexity over the stripe $[3] \times [n]$ in Theorem~\ref{thm: [3]x[n]}. 

\subsection{Overview of the proof}

The lower bounds in Theorems~\ref{thm: [3]x[n]} and~\ref{Thm:HighDimLower} are both obtained using the same general construction. We describe it in the setting of functions over $[n]^d$ for simplicity.

The key idea is that we can construct convex functions whose increase in slope (i.e., second derivative) is small in a particular direction and large in the rest of the directions. We can perturb the values of such functions by $\pm 1$ in a way that yields functions which are far from convex but for which the only violation of convexity on the hypergrid will contain at least two points that form a line along the direction where the slope was increasing slowly. So any algorithm that does not query two points which give a line in that direction cannot catch any violations of convexity. To get a strong lower bound from this key idea, we show that it is possible to ``hide'' the slowly-increasing direction among $\Omega\big( (\frac nd)^d\big)$ possible directions. Since a set of $q$ queries contains pairs of points that form at most $q^2$ different directions, this construction shows that any non-adaptive convexity testing algorithm with one-sided error---i.e., that always accepts convex functions---must have query complexity at least $\Omega\big( (\frac nd)^{d/2}\big)$.

To generalize this argument in a way that gives a lower bound for non-adaptive testing algorithms with two-sided error as well, we consider a different perturbation of the convex functions of $\pm 1$ that preserves convexity.
We can do this by performing the same perturbation (i.e., either all $+1$ or all $-1$) for every point along a line in the slowly-increasing direction. The perturbations for each line are chosen independently at random; by ensuring that the slope of the original function is large enough in all other directions, these independent perturbations do not violate convexity. As we show in the rest of this section, non-adaptive algorithms with query complexity $o\big( (\frac nd)^{d/2}\big)$ cannot distinguish this type of perturbation from the type that breaks convexity.

\subsection{Preliminaries}

We write $x_{[a, b]}$ to denote the coordinates $x_a, x_{a+1},\ldots, x_b$ of an input $x$. We use the following standard results in our proof.

\begin{lem}[Hoeffding's inequality]
	Let $x_1,\ldots, x_n \in \R$ be negatively correlated random variables bounded by $x_i \in [b_i, a_i]$ and define $\overline{x} = \frac{1}{n}(x_1 + \cdots + x_n)$. Then
\[
\Pr\left[|\overline{x} - \E[\overline{x}]| \geq t\right] 
\leq 2e^{\frac{2n^2t^2}{\sum_{i=1}^{n}(b_i - a_i)^2}}.
\]
\end{lem}

\begin{lem}[Yao's minimax]
\label{lem: yao's minimax}
Fix any disjoint sets $\cP$ and $\cN$ of functions mapping $\cX$ to $\cY$.
Let $\cD_{\cP}$ and $\D_{\cN}$ be probability distributions 
on functions mapping $\cX$ to $\cY$ that satisfy 
\[
\Pr_{f \sim \cD_{\cP}}[ f \in \cP ] = 1 
\qquad \text{and} \qquad
\Pr_{g \sim \D_{\cN}}[ f \in \cN ] = 1 - o(1).
\]
Let $\cD$ be the distribution where with probability $\frac{1}{2}$ we sample from $\cD_Y$ and with probability $\frac{1}{2}$ we sample from $\cD_N$. If any non-adaptive deterministic algorithm $\Pi$ with query complexity $q$ can not answer correctly with probability $\frac{2}{3}$, then any non-adaptive randomized algorithm that decides whether $f \in \cP$ or $f \in \cN$ with error at most $\frac{1}{4}$ makes $\Omega(q)$ queries.
\end{lem}

\begin{prop}[Theorem 332 \cite{hardy1979}]
\label{prop: co-prime_size}
Let $a, b \in [n]$ be two numbers picked uniformly at random. The probability that the pair $(a, b)$ is co-prime is $>0.5$. 
\end{prop}

\subsection{Change of basis and convexity}


\begin{defn}
A \emph{lattice basis} is a matrix $B = [b_1, \ldots, b_k] \in \mathbb{R}^{d \times k}$ whose columns are linearly independent vectors in $\mathbb{R}^d$. The \emph{lattice generated by $B$} is the set 
\[
\mathcal{L}(B) = 
\Big\{ B x \mid x \in \mathbb{Z}^k \Big\} =
\Big\{ \sum_{i = 1}^k x_i b_i \mid x_1,\ldots,x_k \in \mathbb{Z} \Big\}.
\]
\end{defn}

\begin{fact}[Lemma 1.2 \cite{Rothvoss2015IntegerOA}]
\label{Fact: basis of Z^d}
$B \in \mathbb{Z}^{[d] \times [d]}$ is a basis of $\mathbb{Z}^d$ if and only if its determinant is $\pm1$. 
\end{fact}




\begin{defn}
Given any vector $a \in \mathbb{Z}^d$ whose first two coordinates $a_1$ and $a_2$ are coprime, the \emph{canonical basis completion} of $a$ is the basis $B(a) = [b_1(a),\ldots, b_d(a)] \in \mathbb{Z}^{d \times d}$ whose $i$th column is 
\[
b_i(a) = \begin{cases}
a & \mbox{if } i = 1 \\
c_1 e_1 + c_2 e_2 & \mbox{if } i = 2 \\
e_i & \mbox{if } 3 \le i \le d
\end{cases}
\]
where $c_1$ and $c_2$ are the integers that satisfy $a_1c_1 - a_2c_2 = 1$ and $e_i \in \mathbb{Z}^d$ is the vector with value $1$ in the $i$th coordinate and $0$ in all other coordinates.
\end{defn}

The next proposition shows that the canonical basis completion of any vector $a \in \mathbb{Z}^d$ that satisfies the condition of the above definition generates the lattice $\mathbb{Z}^d$.

\begin{prop}
\label{prop: Basis completion}
Given any vector $a \in \mathbb{Z}^d$ whose first two coordinates $a_1$ and $a_2$ are coprime, the canonical basis completion $B(a)$ of $a$ generates the lattice $\mathcal{L}(B(a)) = \mathbb{Z}^d$.
\end{prop}

\begin{proof}
Follows from Fact~\ref{Fact: basis of Z^d}.
\end{proof}

If $x \in \mathbb{Z}^d$ be the representation of a point according to the basis $I$, then $x^{B} = B^{-1}x$ is the representation according to the basis $B$. So $y = x + a$ and $y^B = x^B + e_1$ are equivalent.

\subsection{Constructions}

In this subsection we show how to construct the distributions $\mathcal{D}_{Y}, \mathcal{D}_{N}$. We also prove that every function in $\mathcal{D}_{Y}$ is convex and every function in $\mathcal{D}_{N}$ is $\frac{1}{20}$-far from convex. 

Let $\mathcal{B}$ be the distribution over bases obtained by drawing a vector $a \in \mathbb{Z}^d$ uniformly at random among all vectors whose coordinates are in the range $0 \le a_1,a_2,\ldots,a_d \le \frac{n}{4d}$ and whose first two coordinates $a_1$ and $a_2$ are coprime and returning the canonical basis $B(a)$ for $a$. 

The distributions $\mathcal{D}_Y$ and $\mathcal{D}_N$ are both obtained by drawing a basis from $\mathcal{B}$ and starting with a convex function $g_B$ associated with that basis that we will call the \emph{canonical} convex function for $B$.

\begin{defn}
The \emph{canonical convex function} for a basis $B$ of $\mathbb{Z}^d$ is the function $g_B \colon \mathbb{Z}^d \to \mathbb{Z}$ defined by
\[
g_B(x) = (x^{B}_1)^2 + 2 \sum_{i=2}^d (x^{B}_i)^2.
\]
\end{defn}

Our distribution on convex functions is obtained by shifting the values of the canonical convex function $g_B$ in a way that preserves convexity.

\begin{defn}[$\mathcal{D}_Y$]
Let $\mathcal{S}^B$ to be the distribution on functions $h \colon [n]^d \to \mathbb{Z}$ obtained by drawing values $\sigma(z) \in \{\pm 1\}$ independently and uniformly at random for each $z \in \mathbb{Z}^{d-1}$ and defining
\[
h(x) = g_B(x) + \sigma(x^{B}_2,\ldots,x^{B}_d)
\]
for each $x \in [n]^d$. Let $\mathcal{D}_Y$ be the distribution obtained by drawing $B \sim \mathcal{B}$ and then drawing a function $h \sim \mathcal{S}^B$.
\end{defn}

Our distribution on functions that are far from convex is similar, except that the shifts of the canonical convex function $g_B$ are now constructed in a way that will create many disjoint violations of convexity.

\begin{defn}[$\mathcal{D}_N$]
Let $\mathcal{A}^B$ be the distribution on functions $h \colon [n]^d \to \mathbb{Z}$ obtained by drawing values $\sigma(z) \in \{\pm 1\}$ independently and uniformly at random for each $z \in \mathbb{Z}^{d-1}$ and defining
\[
h(x) = g_B(x) + \sigma(x^{B}_2,\ldots,x^{B}_d) \cdot (-1)^{x^{B}_1}
\]
for each $x \in [n]^d$. Let $\mathcal{D}_N$ be the distribution obtained by drawing $B \sim \mathcal{B}$ and then drawing a function $h \sim \mathcal{A}^B$.
\end{defn}
 



We complete this section by showing that the functions in the support of $\mathcal{D}_Y$ are indeed convex and that the functions in the support of $\mathcal{D}_N$ are far from convex.

\begin{clm}
\label{clm: Yes distribution is convex HighDim}
Every function in the support of $\mathcal{D}_Y$ is convex.
\end{clm}

\begin{proof}
Fix any $B$ in the support of $\mathcal{B}$, any $h$ in the support of $\mathcal{S}^B$, and any points $z, x_1,\ldots,x_k \in \mathbb{Z}^d$ such that $z = \sum_{i=1}^k \lambda_i x_i$ is a convex combination of the points $x_1,\ldots,x_k$, $\lambda_1,\ldots,\lambda_k \ge 0$ and $\sum_{i=1}^k \lambda_i = 1$.
We will show that $\sum_{i=1}^k \lambda_i h(x_i) \ge h(z)$.

Let us define $\delta_1,\ldots,\delta_k \in \mathbb{Z}^d$ to be the vectors for which $x^B_i = z^B + \delta_i$ 
for each $i \in [k]$. Then the identity $\sum_{i=1}^k \lambda_i (x_i^{B} - z^B) = 0$ implies that $\sum_{i=1}^k \lambda_i \delta_{ij} = 0$ for every $j \in [d]$ and that
\begin{align*}
\sum_{i=1}^k \lambda_i g_B(x_i) 
&= \sum_{i=1}^k \lambda_i \Big( (x^B_{i1})^2 + 2\sum_{j=2}^d (x^B_{ij})^2 \Big) \\
&= \sum_{i=1}^k \lambda_i \Big( (z^B_1 + \delta_{i1})^2 + 2\sum_{j=2}^d (z^B_j + \delta_{ij})^2 \Big) \\
&= g_B(z) + \sum_{i=1}^k \lambda_i \Big( \delta_{i1}^2 + 2 \sum_{j=2}^d \delta_{ij}^2\Big).
\end{align*}
Define $I = \{i \in [k] \mid z^B_{[2,d]} \neq x^B_{i[2,d]}\}$. For each $i \in I$, the vector $\delta_i$ satisfies $\sum_{j=2}^d \delta_{ij}^{\,2} \ge 1$ so we have that
\[
\sum_{i=1}^k \lambda_i g_B(x_i) - g_B(z) 
\ge 2\sum_{i = 1}^k \lambda_i \sum_{j=2}^d \delta_{ij}^2
\ge 2 \sum_{i \in I} \lambda_i.
\]
Furthermore, since $\sigma(x^B_{i[2,d]}) - \sigma(z^B_{[2,d]})$ is always bounded below by $-2$ and the difference is zero whenever $i \notin I$, we obtain
\[
\sum_{i=1}^k \lambda_i h(x_i) - h(z) \ge
\sum_{i=1}^k \lambda_i g_B(x_i) - g_B(z) - 2\sum_{i \in I} \lambda_i \ge 0. \qedhere
\]
\end{proof}

\begin{clm}
\label{clm: No distributionis far from convex HighDim}
Every function in the support of $\mathcal{D}_N$ is $\frac{1}{20}$-far from convex.
\end{clm}

\begin{proof}
Fix any $B$ in the support of $\mathcal{B}$ and any $h$ in the support of $\mathcal{A}^B$. For any points $x,y,z \in [n]^d$ that satisfy $y^B = x^B + e_1$ and $z^B = y^B + e_1$, if we have
\[
h(x) = g_B(x) - 1, \qquad h(y) = g_B(y) + 1, \mbox{ and } \qquad h(z) = g_B(z) - 1
\] 
Then the triple $(x,y,z)$ is a witness of non-convexity of $h$ since 
\[
\tfrac12 h(x) + \tfrac12 h(z) = \tfrac12 g_B(x) + \tfrac12 g_B(z) - 1 = g_B(y) < h(y) = h(\tfrac12 x + \tfrac12 z).
\]

Hence from how we defined $h$, any four points $w, x, y, z \in [n]^d$ that satisfy $x^B = w^B + e_1$, $y^B = x^B + e_1$ and $z^B = y^B + e_1$ one of $(w, x, y)$, $(x, y, z)$ is a witness on non-convexity. Let $L = [\frac{n}{2d}, n - \frac{n}{2d}]^d$. For $s \in \mathbb{Z}^{d-1}$, let $L_s = \{x \mid x \in [n]^d, \exists y \in L \text{ s.t } y^{B}_{[2..d]} = x^{B}_{[2..d]} = s\}$. Since $a_1, a_2,..., a_d < \frac{n}{4d}$ we have that $|L_s| \geq 4$. And since any $4$ consecutive points with the same $[2, d]$ coordinates, in basis $B$, have a witness of non-convexity, the number of witnesses in $L_s$ is $\geq \frac{|L_s|}{7}$. Also $L \subseteq \cup_{s \in \mathbb{Z}^{d-1}} L_s$, hence the number of disjoint witnesses of non-convexity is greater than $\frac{|L|}{7} = \frac{1}{7}(1 - \frac{1}{d})^dn^d \geq \frac{1}{20} n^d $. In every disjoint non-convexity witness we have to change the value of at least one point to make the function convex. Therefore $h$ is $\frac{1}{20}$-far from convex. 
\end{proof}

\subsection{Proof of Theorem~\ref{Thm:HighDimLower}}

Let $\cD$ be the distribution where with probability $\frac{1}{2}$ we pick something from $\cD_Y$ and with probability $\frac{1}{2}$ we pick something from $\cD_N$. In this section we prove that there does not exist a non-adaptive deterministic algorithm with query complexity $q< 0.01 (\frac{n}{4d})^{\frac{d}{2}}$ that answers correctly with probability $\frac{2}{3}$ on the distribution $\cD$. From Lemma~\ref{lem: yao's minimax} this would prove Theorem~\ref{Thm:HighDimLower} as from Claim~\ref{clm: Yes distribution is convex HighDim} and Claim~\ref{clm: No distributionis far from convex HighDim} we know that every function in the support of $\mathcal{D}_Y$ is convex and every function in the support of $\mathcal{D}_N$ is $\frac{1}{10}$-far from convex.

Let us assume there exists such a deterministic algorithm $\Pi$ that answers correctly on a distribution $\cD = \frac{1}{2}\cD_Y + \frac{1}{2}\cD_N$ with probability greater than $\frac{2}{3}$. We can think of the distribution $\cD$ as pick a $B \sim \mathcal{B}$ and pick a $\sigma\colon \mathbb{Z}^{d-1} \rightarrow \pm 1$ uniformly at random. And at the end with probability $\frac{1}{2}$ we choose whether we want a function in the support of $\cD_Y$ or $\cD_N$. Let the points the algorithm $\Pi$ queries be $Q = x_1, x_2, ...., x_q \in \mathbb{Z}^d$. 

We refer to a $B$ in the support of $\mathcal{B}$ to be \emph{exposed} if there exists $i, j <q$ such that $x^B_{i[2,d]} = x^B_{j[2, d]}$, otherwise we refer to it as \emph{hidden}. 

\begin{clm}
\label{clm: non-adaptive deterministic lower bound HighDim}
On the distribution $\cD$ the probability that $\Pi$ answers correctly is less than $0.6$.
\end{clm}
\begin{proof}
When $B$ is hidden then there is no way the algorithm $\Pi$ can answer correctly with probability greater than $\frac{1}{2}$. This is because $\Pr_{f \sim \mathcal{D}_{Y}|_{B \text{ is hidden}}}\left[ f|_Q = \alpha \right] 
= \Pr_{g \sim \mathcal{D}_{N}|_{B \text{ is hidden}}}\left[ g|_Q = \alpha \right]$. In fact it is even stronger, along with function values at the queried points even if we give what the hidden basis $B$ is, the algorithm can not answer correctly with probability greater than $\frac{1}{2}$. This is because, as for any $i, j <q$,  $x^B_{i[2, d]} \neq x^B_{j[2,d]}$, we have  $f|_Q - g_B|_Q = s$, for each $s \in \{-1, +1\}^q$, with probability $\frac{1}{2^q}$ irrespective of $f$ being in $\cS^B$ or $\cA^B$. We can assume that the algorithm always answers correctly when $B$ is exposed. The probability that the algorithm $\Pi$ answers correctly is $\leq \Pr[B \text{ is exposed}]\cdot 1 + \Pr[B \text{ is hidden}]\cdot \frac{1}{2}$.

Since there are only $\binom{q}{2}$ $< q^2$,  $\text{ }i, j <q$ pairs, there are at most $q^2$ exposed $B$. From Proposition~\ref{prop: co-prime_size} and the construction of $\mathcal{B}$ we know that $|\mathcal{B}| \geq 0.5 (\frac{n}{4d})^d$ and if $q< 0.01 (\frac{n}{4d})^{\frac{d}{2}}$ the probability that a $B \sim \mathcal{B}$ is exposed is $\leq \frac{1}{100}$.

Hence the success probability of the algorithm is $\leq \frac{1}{100}\cdot 1 + \frac{99}{100} \cdot \frac{1}{2} \leq \frac{101}{200}$.
\end{proof}

This a contradiction on the assumption that the algorithm answers correctly with probability $\frac{2}{3}$. Hence there can not exist such a non-adaptive deterministic algorithm $\Pi$.

\subsection{Non-adaptive lower bound for \texorpdfstring{$[3] \times [n]$}{[3]x[n]}}

\label{sec: [3]x[n] lower bound}

In this section we prove a $\Omega(\sqrt{n})$ lower bound for non-adaptively testing convexity on the $[3]\times [n]$ grid. The proof is almost the same as the the higher dimensional setting with slight changes. 

Let $\mathcal{B}$ be the distribution over bases obtained by drawing a vector $a \in \mathbb{Z}^2$ uniformly at random among all vectors whose first coordinate is $1$ and the second coordinate is in the range $0 \leq a_2 \leq \frac{n}{100}$ and returning the canonical basis $B(a)$ for $a$. 

Define the distributions $\mathcal{D}_{Y}$ and $\mathcal{D}_{N}$ as above with the one modification that the domain of $h$ is set to be $[3] \times [n]$ instead of $[n]^d$. In this setting, we again have that every function in the support of $\mathcal{D}_Y$ is convex, using the same argument as in Claim~\ref{clm: Yes distribution is convex HighDim}.  But now it is no longer true that every function in the support of $\mathcal{D}_N$ is $\frac{1}{10}$-far from convex. Instead, we have that a function $f \sim \mathcal{D}_N$ is $\frac{1}{10}$-far from convex with probability $1-o(1)$.


\begin{clm}
A function $f \sim \mathcal{D}_N$ is $\frac{1}{10}$-far from convex with probability $1 - o(1)$.
\end{clm}

\begin{proof}
For any $B$ in the support of $\mathcal{B}$, a function $h \sim \mathcal{A}^B$ is $\frac{1}{10}$-far from convex with probability $1-o(1)$. Let $X= \{x \mid x_1 = 0, 0 \leq x_2 \leq \frac{9n}{10}\}$. For any points $x \in X$ and $y, z \in \mathbb{Z}^2$ that satisfy $y^B = x^B + e_1$ and $z^B = y^B + e_1$ we have that $y, z \in [3] \times [n]$ and 
\[
h(x) = g_B(x) - 1, \qquad h(y) = g_B(y) + 1, \mbox{ and } \qquad h(z) = g_B(z) - 1
\]
with probability $\frac{1}{2}$. Therefore, $x, y, z$ form a witness for non-convexity with probability $\frac{1}{2}$. This is true for all $x \in X$. Using Hoeffding's inequality the probability that the number of witnesses for non-convexity is less than $\frac{n}{3}$ is $\leq  e^{-cn}$. Hence with probability $1 - e^{-cn}$ the distance to convexity is at least $\frac{\frac{n}{3}}{3n} \geq \frac{1}{10}$.
\end{proof}

Any non-adaptive deterministic algorithm which performs $q < \frac{\sqrt{n}}{100}$ can not answer correctly with probability grater than $0.6$. The proof is similar to that of Claim~\ref{clm: non-adaptive deterministic lower bound HighDim}. From Lemma~\ref{lem: yao's minimax} this completes the proof of the lower bound in Theorem~\ref{thm: [3]x[n]}.

\section{Lower bound for testing convexity on the line}
\label{sec: 1-D lower bound}

The lower bound in Theorem~\ref{thm:1D} is obtained by using similar ideas to the ones in~\cite{Belovs18} used to prove the analogous lower bound for testing monotonicity. The key idea is to introduce violations of convexity that are only visible at a given scale.

We first show a lower bound of $\Omega(\log n)$ for $\frac{1}{9}$-testing convexity and then extend it to general $\epsilon$.

\subsection{General principle}

In this section, we formulate the general principle our proof is based on in an abstract form to give the overall structure of our proof.
In the next sections, we show how to apply it to convexity testing.

We deal with randomised query algorithms whose inputs are functions $f\colon [n]\to[r]$, and which want to distinguish the set of positive inputs $\cP$ from the set of negative inputs $\cN$, that is, accept all $f\in\cP$ and reject all $g\in\cN$.
If $\cT$ is a deterministic decision tree, then $\cT(f)$ denotes the terminal leaf of the decision tree $\cT$ on input $f$.

\begin{lem}
\label{lem:Yao-easy}
Let $\cP$ and $\cN$ be two disjoint sets of functions mapping $[n]$ to $[r]$.
Let $A$ and $B$ be sets of labels, and assume there are mappings 
$A\ni a \mapsto f_a \in \cP$ and
$B\ni b \mapsto g_b \in \cN$.
Let $\mu$ and $\nu$ be two probability measures supported on $A$ and $B$, respectively.

Assume that for every deterministic decision tree $\cT$ of depth $q$, one can find a partial mapping $\eta\colon B\to A$ such that
\begin{itemize}
\item $\cT(f_{\eta(b)}) = \cT(g_b)$ for every $b$ in the domain of $\eta$;
\item $\mu(\eta(B)) = \Omega(1)$;
\item $\nu(\eta^{-1}(a)) = \Omega(\mu(a))$ for every $a\in\eta(B)$.
\end{itemize}
Then, every randomised query algorithm distinguishing $\cP$ from $\cN$ makes $\Omega(q)$ queries.
\end{lem}

\begin{proof}
Assume $\mu(\eta(B)) \ge C_1$ and $\nu(\eta^{-1}(a)) \ge C_2 \mu(a)$ for every $a\in\eta(B)$, where $C_1, C_2 >0$ are constants.
Performing standard error reduction, we may assume that the error probability of the algorithm is a constant $\eps>0$, which depends on $C_1$ and $C_2$ in a way to be determined later.

By standard Yao's principle, there exists a deterministic decision tree $\cT$ that accepts with probability $\ge 1-2\eps$ on $f_a$ where $a\sim\mu$, and rejects with probability $\ge 1-2\eps$ on $g_b$ where $b\sim\nu$.

Let $P$ be the set of $a\in A$ such that $\cT$ accepts $f_a$.
By the first property of $\eta$, $\cT$ accepts all $g_b$ with $b\in \eta^{-1}(P)$.
We have
\[
\nu(\eta^{-1}(P)) \ge C_2 \mu(P\cap \eta(B)) \ge C_2(C_1 - 2\eps).
\]
As this quantity is supposed to be less then $2\eps$, we get a contradiction when
$2\eps < C_2(C_1 -2\eps)$, 
or
\[
\eps < \frac{C_1 C_2}{2(1+C_2)},
\]
which is a positive constant.
\end{proof}

\subsection{The case of \texorpdfstring{$\eps = \Omega(1)$}{epsilon = Omega(1)}}
\label{sec:Omega1}
In this subsection we prove the following theorem, which covers the $\eps=\Omega(1)$ case of Theorem~\ref{thm:1D}.

\begin{thm}
\label{thm:line-lb}
For an integer $k$, it takes $\Omega(k)$ queries to $\frac19$-test a function $f\colon [3^k]\to [(9k)^{3k}]$ for convexity.
\end{thm}

We will define the required objects from Lemma~\ref{lem:Yao-easy}.
Clearly, $n=3^k$ and $r = (9k)^{3k}$.
The sets $\cP$ and $\cN$ consist of convex and $1/9$-far-from-convex functions, respectively.

\newcommand{\TK}{[3]^{<k}}
Let $m=3k^3$.
Denote by $\TK$ the set of ternary strings of length strictly less than $k$, including the empty string.
The set $A$ consists of all the functions from $\TK$ into $[k^3-1]$.  For $a\in A$, the value of $a$ on $s\in\TK$ is denoted by $a_s$.
We define the function $f_a\in\cP$ corresponding to $a\in A$ by giving its discrete derivative, which is a monotone function $\partial f_a\colon [3^k]\to[m^k]$.
That is,
\begin{equation}
\label{eqn:f}
f_a(x) = \sum_{z<x} \partial f_a(z).
\end{equation}
It is clear that if the function $\partial f_a$ is monotone, the function $f_a$ is convex.
Also, the maximal value of $f_a$ is at most $3^k\cdot m^k < (9k)^{3k}$.

The function $\partial f_a$ is defined as follows.
Assume that the argument $x\in [3^k]$ is written in ternary and the value $\partial f_a(x)\in [m^k]$ in $m$-ary.
We prepend leading zeroes if necessary so that each number has exactly $k$ digits.
We enumerate the digits from left to right with the elements of $[k]$, so that the $0$-th digit is the most significant one, and the $(k-1)$-st digit is the least significant one.
We use $x_i$ to denote the $i$th digit of $x$.
For an interval $[a..b]$, we define $x_{[a..b]}$ as the substring of $x$ formed by the digits $x_i$ as $i$ ranges over $[a..b]$.

Let
\begin{equation}
\label{eqn:phi}
\phi_a(x,i) = \begin{cases}
a_{x_{[i]}} & \mbox{if } x_i = 0; \\
a_{x_{[i]}}+1 & \mbox{if } x_i = 1; \\
m - 2a_{x_{[i]}}-1 & \mbox{if } x_i = 2; \\
\end{cases}
\end{equation}
for $x\in [3^k]$ and $i\in [k]$.
The $i$-th digit of $\partial f_a(x)$ is equal to $\phi(x,i)$.
That is,
\begin{equation}
\label{eqn:partial f}
\partial f_a(x) = \sum_{i=0}^{k-1} m^{k-1-i} \phi_a(x,i).
\end{equation}

Let us make some clarifying comments here.
The main case of interest in Equation~(\ref{eqn:phi}) is $x_i=0$ and $x_i=1$.  In the far-from-convex case, the first two cases will be essentially switched.
This makes the function far from convex, but it is hard to see that just by observing the $x_i=0$ or $x_i=1$ case independently.
The $x_i=2$ case is necessary to ensure that the sum of the elements on the right-hand side of Equation~(\ref{eqn:phi}) is independent of $a_s$, see Claim~\ref{clm: dependence of f on a_s}.

\begin{clm}
\label{clm:fa_convex}
Every function $f_a$ is convex.
\end{clm}

\pfstart
Every function $\partial f_a$ is monotone because $a_s<a_s+1<m-2a_s-1$ for every $s \in \TK$.
\pfend

\begin{clm}
\label{clm: dependence of f on a_s}
The value $f_a(x)$ only depends on the values of $a_s$ as $s$ ranges over the prefixes of $x$, and is independent from the remaining values of $a_s$.
\end{clm}
\begin{proof}
From Equation~(\ref{eqn:f}) and Equation~(\ref{eqn:partial f}), we can write
\begin{equation}
\label{eqn:dep_f(x)}
f_a(x) = \sum_{i=0}^{k-1} m^{k-1-i} \sum_{z<x} \phi_a(z,i).
\end{equation}
Note that the sum of the elements on the right-hand-side of Equation~(\ref{eqn:phi}) is $m$ for every value of $a_{x[i]}$.
This means that for every $s\in[3]^{i-1}$ such that $s < x_{[i]}$, we have
\[
\sum_{z<x: z_{[i]}=s} \phi_a(z,i) = 3^{k-1-i}\cdot m.
\]
The number of such $s$ is exactly $x_{[i]}$.
Using this, and summing explicitly over $z<x: z_{[i]} = x_{[i]}$, we get that
\[
f(x) = \sum_{i=0}^{k-1} (3m)^{k-1-i}m x_{[i]} + \sum_{i=0}^{k-1} m^{k-i-1} \cdot
\begin{cases}
x_{[i+1..k]} \cdot a_{x_{[i]}}, & \mbox{if } x_i = 0; \\
3^{k-1-i}\cdot a_{x_{[i]}} + x_{[i+1..k]} (a_{x_{[i]}}+1), & \mbox{if } x_i = 1; \\
3^{k-1-i} (2a_{x_{[i]}}+1) + x_{[i+1..k]} (m-2a_{x_{[i]}}-1), & \mbox{if } x_i = 2. \\
\end{cases}
\]
\end{proof}


The set $B$ is defined as $A\times [k]$.
For $a\in A$, $j\in [k]$, and $\delta=\pm 1$, let $a[j,\delta]$ denote the function $b\colon\TK\to [-1..k^3]$ defined by
\[
b_s = 
\begin{cases}
a_s+\delta,&\text{if $|s|=j$;}\\
a_s,&\text{otherwise.}
\end{cases}
\]
Note that the value of $b_s$ may lie outside of $[k^3-1]$, but the definition $f_b$ still makes sense, and Claim~\ref{clm: dependence of f on a_s} still holds.

For $(a,j)\in B$, the corresponding function $g_{a,j}$ is defined by
\[
g_{a,j}(x) = 
\begin{cases}
f_{a[j,+1]}(x),&\text{if $x_j=0$;}\\
f_{a[j,-1]}(x),&\text{if $x_j=1$;}\\
f_{a}(x),&\text{if $x_j=2$.}\\
\end{cases}
\]
%
%

\begin{clm}
\label{clm:far-convex}
Every function $g_{a,j}$ is $\frac19$-far from convex.
\end{clm}

\begin{proof}
First consider the case $j < k-1$.
Partition the domain of $g_{a,j}$ into $9$-tuples which differ only in the $j$th and the last $(k-1)$th ternary digits. 
In a given $9$-tuple, let $x$ and $y$ be the inputs that satisfy $x_j = 0, x_{k-1} = 0$ and $y_j = 1, y_{k-1} = 0$. 
The definition of $g_{a,j}$ implies that 
\[
\partial g_{a,j}(x) = m^{k-1-j} (a_{x_{[j]}}+1) + \sum_{i = 0}^{j-1} m^{k-1-i} \phi_a(x,i) + 
\sum_{i = j+1}^{k-1} m^{k-1-i} \phi_a(x,i)
\]
and
\[
\partial g_{a,j}(y) = m^{k-1-j} a_{y_{[j]}} + \sum_{i = 0}^{j-1} m^{k-1-i} \phi_a(y,i) + 
\sum_{i = j+1}^{k-1} m^{k-1-i} \phi_a(y,i).
\]
Since $y_{[j]} = x_{[j]}$, we have $a_{y_{[j]}} = a_{x_{[j]}}$ and $\phi_a(y,i) = \phi_a(x,i)$ for each $i < j$. 
Therefore,
\begin{align*}
\partial g_{a,j}(y) - \partial g_{a,j}(x) 
&= - m^{k-1-j} + \sum_{i = j+1}^{k-1} m^{k-1-i} (\phi_a(y,i) - \phi_a(x,i)) \\
&\le -m^{k-1-j} + \sum_{i=j+1}^{k-1} m^{k-1-i} (m-1) \\
&= -m^{k-1-j} + (m^{k-1-j} - 1) < 0
\end{align*}
and so $\partial g_{a,j}(y) < \partial g_{a,j}(x)$. Any convex function must disagree with $g_{a,j}$ on at least one of the four points $x$, $x+1$, $y$, or $y+1$.

The case $j=k-1$ is similar, but only considering the triples which differ in the last, $(k-1)$st, digit.
\end{proof}

The probability distributions $\mu$ and $\nu$ are uniform on $A$ and $B$, respectively.

Let $\cT$ be a deterministic decision tree of depth $q\le k/2$.
Now we define the mapping $\eta\colon B\to A$ which depends on $\cT$.

We will define $\eta$ in the inverse direction, starting from a potential image $a\in A$.  
Let $Q = \{x_1,\dots,x_q \}\subseteq [3^k]$ be the values which the decision tree $\cT$ queries on input of $f_a$.
Denote $S = \{x_{[j]} \mid x\in Q,\; j\in[k] \}$.
We will proceed only if 
\begin{equation}
\label{eqn:condition a_s}
a_s\in [1..k^3-2]
\qquad\text{for all}\quad s\in S.
\end{equation}
Take $j\in [k]$, and define $b$ as
\begin{equation}
\label{eqn:b^j}
b_s = \begin{cases}
a_s - 1 & \mbox{if $s = x_{[j]}$ for some $x \in Q$ with $x_j = 0$;} \\
a_s + 1 & \mbox{if $s = x_{[j]}$ for some $x \in Q$ with $x_j = 1$;} \\
a_s & \mbox{if $s = x_{[j]}$ for some $x \in Q$ with $x_j = 2$;} \\
a_s & \mbox{otherwise.}
\end{cases}
\end{equation}
if there are no conflicts among the first three cases in this definition.
Note that Equation~(\ref{eqn:condition a_s}) implies that $b\in A$.
If $b$ is well-defined, we let $\eta(b,j) = a$.

\begin{clm}
\label{clm:eta}
The mapping $\eta$ is well-defined and $\cT(f_a) = \cT(g_{b, j})$ in the above notation.
\end{clm}

\pfstart
By definition of $g_{b,j}$ and $b$, and using Claim~\ref{clm: dependence of f on a_s}, we have that $f_a(x) = g_{b,j}(x)$ for all $x\in Q$.
This proves that $\cT(f_a) = \cT(g_{b,j})$.

Now consider $(b,j)$ in the domain of $\eta$.
By the previous paragraph it can only come from $a\in A$ such that $\cT(f_a) = \cT(g_{b,j})$.  Then, the set $Q$ is known, and the mapping in Equation~(\ref{eqn:b^j}) can be inverted, proving that $\eta$ is well-defined.
\pfend

\begin{clm}
\label{clm:muetaB}
We have $\mu(\eta(B)) = \Omega(1)$ and $|\eta^{-1}(a)| \ge k/2$ for every $a\in\eta(B)$. 
\end{clm}

\pfstart
We will prove first that $a\in\eta(B)$ if condition Equation~(\ref{eqn:condition a_s}) is satisfied.
Indeed, in this case, we do not set $\eta(b,j)=a$ only if there are two inputs $x,y\in Q$ such that $x_{[j]} = y_{[j]}$ and $x_j\ne y_j$.
By a simple modification of \cite[Lemma 6]{Belovs18}, there are at most $|Q|-1$ values of $j$ for which this happens.  As $|Q|\le k/2$, this proves the second part of the claim.

For the first part of the claim, the probability that Equation~(\ref{eqn:condition a_s}) does not hold is upper bounded by the union bound over at most $k/2$ elements of $Q$ and $k$ prefixes of each $x\in Q$ as
\[
\Pr_{a\sim\mu}[\text{Equation~(\ref{eqn:condition a_s}) does not hold}] \le \frac k2\cdot k\cdot\frac{2}{k^3-1} = o(1).\qedhere
\]
\pfend

Now we can apply Lemma~\ref{lem:Yao-easy} and get that complexity of $1/9$-testing functions for convexity is $\Omega(k) = \Omega(\log n)$ as required.

\subsection{General lower bound for the line}

The lower bound can be strengthened for general values of $\epsilon$ as follows.

\begin{thm}
\label{thm:line-lb-eps}
Fix any $\frac1n \le \epsilon \le \frac19$. Any $\epsilon$-tester for convexity of functions $[n] \to \bZ$ has query complexity
\[
\Omega\left(\tfrac1\epsilon \log(\epsilon n)\right)
\]
\end{thm}

The proof of Theorem~\ref{thm:line-lb-eps} is a slight extension of the proof of Theorem~\ref{thm:line-lb}.

Define $\ell = \lceil \frac1{9\epsilon} \rceil$, $k= \lfloor \log_3 \frac{n}{\ell} \rfloor$, and $m = 3k^3$. We will show that $\epsilon$-testing the convexity of a function mapping $[\ell 3^k] \to \bZ$ requires $\Omega(\ell k)$ queries.

We will use notations with tilde for objects referring to the proof of Theorem~\ref{thm:line-lb-eps}, and non-tilde notation for the objects from Section~\ref{sec:Omega1}.

\newcommand{\tx}{t\cdot 3^k+x}
\newcommand{\tg}{\tilde g}
\newcommand{\tpf}{\widetilde{\partial f}}

Let $A$ be as in Section~\ref{sec:Omega1}, and define $\tA = A^\ell$.
For $a\in\tA$, we have $a = (a^0,\dots,a^{\ell-1})$ with each $a^t\in A$.
The partial derivative is given by
\[
\tpf_a(\tx) = t\cdot m^k + \partial f_{a^t}(x),
\]
for $t\in[\ell]$, $x\in[3^k]$, and $\partial f_a$ as in Section~\ref{sec:Omega1}.
The function $\tf_a$ is given by
\[
\tf_a(\tx) = \sum_{z< \tx} \tpf_a(\tx).
\]

Similarly to Claims Claim~\ref{clm:fa_convex} and Claim~\ref{clm: dependence of f on a_s}, we have the following result
\begin{clm}
Every function $\tf_a$ is convex.
The value of $\tf_a(\tx)$ only depends on the values of $a^t_s$ as $s$ runs through the prefixes of $x$.
\end{clm}

The set $\tB$ is defined as $\tA\times[\ell]\times[k]$.
For $a\in\tA$, define $a[t,j,\delta]$ as $b=(b^0,\dots,b^{\ell-1})$ with $b^t = a^t[j,\delta]$ and $b^u = a^u$ for $u\ne t$.
Then,
\[
\tg_{a,t,j}(\tx) =
\begin{cases}
\tf_{a[t,j,+1]}(\tx),&\text{if $x_j=0$;}\\
\tf_{a[t,j,-1]}(\tx),&\text{if $x_j=1$;}\\
\tf_{a}(\tx),&\text{if $x_j=2$.}\\
\end{cases}
\]

\begin{clm}
Every function $\tg_a$ is $\eps$-far from convex.
\end{clm}

\begin{proof}
This is due to the fact that there are $3^{k-2} \ge \epsilon \cdot \ell 3^{k}$ disjoint pairs of values $x < y$ for which $\partial g(y) > \partial g(x)$, as in the proof of Claim~\ref{clm:far-convex}.
\end{proof}

The probability distributions $\mu$ and $\nu$ are defined as uniform on $\tA$ and $\tB$, respectively.

The mapping $\eta\colon \tB\to\tA$ is also defined similarly to Section~\ref{sec:Omega1}.
Let $\cT$ be a deterministic decision tree of depth $q\le \frac{\ell k}{4}$.
Take $a\in\tA$.  
Let $Q$ be the set of variables queried by $\cT$ on $\tf_a$, and let 
$
Q^t = \{x\in[3^k] \mid \tx \in Q\}
$.
For $(t,j)\in[\ell]\times[k]$, let
\begin{equation}
\label{eqn:bts}
b^t_s = \begin{cases}
a^t_s - 1 & \mbox{if $s = x_{[j]}$ for some $x \in Q^t$ with $x_j = 0$;} \\
a^t_s + 1 & \mbox{if $s = x_{[j]}$ for some $x \in Q^t$ with $x_j = 1$;} \\
a^t_s & \mbox{if $s = x_{[j]}$ for some $x \in Q^t$ with $x_j = 2$;} \\
a^t_s & \mbox{otherwise.}
\end{cases}
\end{equation}
and $b=(a^0,\dots,a^{t-1},b^t,a^{t+1},\dots,a^{\ell-1})$.
We call the pair $(t,j)$ good if there are no conflicts in the first three cases of Equation~(\ref{eqn:bts}) (that is, $b^t_s$ is well-defined) and $b\in \tB$.
If there are at least $\ell k/2$ good pairs, we define $\eta(b,t,j) = a$ for each good pair $(t,j)$, where $b$, of course, depends on $t$ and $j$.

Similarly to Claim~\ref{clm:eta}, $\eta$ is well-defined and $\cT(\tf_{\eta(b)}) = \cT(\tg_b)$ for every $b$ in the domain of $\eta$.
Also, by definition, $\nu(\eta^{-1}(a)) = \Omega(\mu(a))$ for every $a\in\eta(\tB)$.
In order to apply Lemma~\ref{lem:Yao-easy}, it remains to show the following.

\begin{clm}
We have $\mu(\eta(B)) = \Omega(1)$.
\end{clm}

\pfstart
Fix $a\in\tA$.
Similarly to Claim~\ref{clm:muetaB}, there can be at most $q-1< \frac{\ell k}{4}$ pairs such that there is a contradiction in the first three cases of Equation~(\ref{eqn:bts}).
A pair $(t,j)$ can be bad also because $a^t_s$ equals $0$ or $k^3-2$.
The expected number of such pairs as $a\sim \tA$ is
$
q\cdot k\cdot \frac{2}{k^3-1} = O(\frac{\ell}{k}).
$
By Markov's inequality, probability that the number of such pairs is $\ge \frac{\ell k}{4}$ is $o(1)$.
And if this does not happen, the number of good pairs is at least $\ell k/2$.
\pfend

\section*{Acknowledgements}
Aleksandrs Belovs is supported by the ERDF grant number 1.1.1.2/VIAA/1/16/113.
Eric Blais and Abhinav Bommireddi are funded by an NSERC Discovery grant.

\bibliographystyle{plain}
\bibliography{TestingConvexity}

\appendix

\section{On convexity and line convexity}
\label{app:otherConvexity}

In the introduction, we mentioned that the notion of \emph{linear convexity} studied in~\cite{BlaisRY14} is not equivalent to the notion of convexity we study in this current work. In this section, we provide a proof of this statement.

\begin{defn}
Fix a set $X \subseteq \R^d$. The function $f\colon X \to \R$ is \emph{linearly convex} if for every $x, y \in X$ and every $0 \le \lambda \le 1$ for which $\lambda x + (1-\lambda)y \in X$, we have $f(\lambda x + (1-\lambda) y) \le \lambda f(x) + (1-\lambda) f(y)$.
\end{defn}

When $X = \R^d$ or, more generally, when $X$ is a convex set, then the notion of linear convexity is equivalent to convexity. When $X$ is a discrete set with dimension $d \ge 2$, however, the two definitions are not equivalent.

\begin{prop}
\label{prop:conv-lin}
For any $d \ge 2$ and any discrete set $X \subseteq \R^d$, every convex function $f\colon X \to \R$ is also linearly convex. However, for every $d \ge 2$ there are discrete sets $X \subseteq \R^d$ for which there exist linearly convex functions $g\colon X \to \R$ that are not convex.
\end{prop}

\begin{proof}
That every convex function $f$ is also linearly convex follows directly from the definitions. For the second statement, consider the function $f \colon [3] \times [3] \to \R$ defined by
\begin{align*}
f(0, 2) = 3 && f(1, 2) = 1 && f(2, 2) = 5 \\ 
f(0, 1) = 1 && f(1, 1) = 2 && f(2, 1) = 3 \\
f(0, 0) = 5 && f(1, 0) = 3 && f(2, 0) = 1 
\end{align*}
The function $f$ is linearly convex, but it has a violation of convexity on the point $(1,1)$ with respect to the points $(2,0)$, $(0,1)$, and $(1,2)$.
\end{proof}

\begin{figure}
\centering
\begin{tikzpicture}[thick,scale=1, every node/.style={transform shape}]
\foreach \x in {0,1,...,2} {
        \foreach \y in {0,1,...,2} {
            \fill[color=black] (\x,\y) circle (0.04);
        }
    }

\node [left] at (0, 0) {$5$};
\node [left] at (0, 1) {$1$};
\node [left] at (0, 2) {$3$};
\node [below] at (1, 0) {$3$};
\node [below] at (1, 1) {$2$};
\node [above] at (1, 2) {$1$};
\node [right] at (2, 0) {$1$};
\node [right] at (2, 1) {$3$};
\node [right] at (2, 2) {$5$};

\draw (2, 0) -- (0, 1) -- (1, 2) -- (2, 0);

\end{tikzpicture}
\caption{Illustration of the function $f$ constructed in the proof of Proposition~\ref{prop:conv-lin}.}
\end{figure}
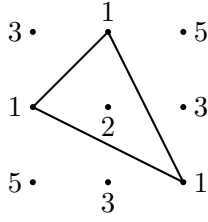

We note that many other notions of convexity of functions over discrete domains have also been considered in the context of discrete convex analysis. See ~\cite{Murota98a} and the references therein for more details on those notions.

\section{Missing proofs from Section~\ref{sec:convexity}}
\label{app:missingProofs}

For completeness, we include proofs of Lemma~\ref{lem:extension} and Theorem~\ref{thm:minimalSimplex} in this section.

\subsection{Proof of Lemma~\ref{lem:extension}}

By convexity of $f$, we have that $g(x) = f(x)$ for all $x\in X$. Hence, $g$ indeed extends $f$.  It remains to prove that $g$ is convex.

\begin{clm}
\label{clm:gAlternative}
The definition of $g$ in Equation~(\ref{eqn:g}) does not change if we minimise over all possible convex combinations $z = \lambda_1 x_1 +\cdots + \lambda_k x_k$, where $x_1,\dots,x_k$ need not form a simplex.
\end{clm}

\pfstart
Let $g(z)$ be defined as in the statement of this claim.
Take a linear combination $z = \lambda_1 x_1 +\cdots + \lambda_k x_k$ which minimises $\sum_i \lambda_i f(x_i)$ and such that $k$ is as small as possible.
We claim that then $x_1,\dots,x_k$ form a simplex.

Indeed, assume $x_1,\dots, x_k$ are not affinely independent.
Then, there exists a non-trivial linear combination $\beta_1 x_1+\cdots+\beta_k x_k = 0$ such that $\beta_1+\cdots+\beta_k = 0$.
Changing the sign of each $\beta_i$ if necessary, we may assume that $\beta_1 f(x_1)+\cdots+\beta_k f(x_k)\ge0$.  
Let $t\ge 0$ be the maximal real number such that $\lambda_i - t \beta_i \ge 0$ for all $i$.

Let $\lambda_i' = \lambda_i - t\beta_i$.
We have that $\lambda_i' \ge 0$, $\sum_i \lambda_i' = 1$, $\sum_i \lambda_i' x_i = z$, and $\sum_i \lambda_i' f(x_i) \le \sum_i \lambda_i f(x_i)$.
Moreover, at least one of $\lambda_i'$ is equal to 0, which contradicts minimality of $k$.
\pfend

\begin{clm}
The function $g$ is convex on the convex hull of $X$.
\end{clm}

\pfstart
Consider a convex combination $z = \mu_1 z_1 + \cdots + \mu_k z_k$, where all $z_i$ lie in the convex hull of $X$.
For each $z_i$ choose a convex combination $z_i = \sum_{x\in X} \lambda_{i,x} x$ such that $g(z_i) = \sum_{x\in X} \lambda_{i,x}f(x)$.
Then, $\lambda_x = \sum_i \mu_i \lambda_{i,x}$ give a convex combination over the elements of $X$ such that $z = \sum_{x\in X} \lambda_x x$.
By Claim~\ref{clm:gAlternative},
\[
g(z) \le \sum_{x\in X} \lambda_x f(x) = \mu_1 g(z_1) + \cdots + \mu_k g(z_k),
\]
proving that the function $g$ is convex.
\pfend

By~\cite{yan2012extension}, the function $g$ can be extended from the convex hull of $X$ to the whole $\bR^d$.  This completes the proof of Lemma~\ref{lem:extension}.

\subsection{Proof of Theorem~\ref{thm:minimalSimplex}}

Assume that $f$ is not convex.
Then there exists a convex combination $z = \lambda_1 x_1 +\cdots+\lambda_k x_k$ such that
$f(z) > \lambda_1 f(x_1) +\cdots+\lambda_k f(x_k)$.
Choose such a convex combination that $k$ is as small as possible and the convex hull of $x_1,\dots,x_k$ is inclusion-wise minimal.
We claim that then $x_1,\dots,x_k,z$ form a minimal centred simplex.

Using the same argument as in Claim~\ref{clm:gAlternative}, we get that $x_1,\dots,x_k$ is a simplex.  
Assume it contains more than two points in its convex hull minus the vertices.
Let $z$ be such that the violation $f(z) - \lambda_1 f(x_1) -\cdots-\lambda_k f(x_k)>0$ is as large as possible.
Let $y$ be any other point in the convex hull of $x_1,\dots,x_k$ except for its vertices.
Then,
\begin{equation}
\label{eqn:zViolation}
f(z) - \lambda_1 f(x_1) -\cdots-\lambda_k f(x_k) \ge f(y) - \mu_1 f(x_1) -\cdots-\mu_k f(x_k),
\end{equation}
where $y = \mu_1 x_1 + \cdots \mu_k x_k$.
Let $t\ge 0$ be the largest real number such that $\lambda_i - t\mu_i\ge 0$ for all $i$.
Let $\lambda_i' = \lambda_i - t\mu_i$.
We have the following convex combination:
\[
z = ty + \lambda_1' x_1 + \cdots + \lambda_k' x_k.
\]
Moreover, one of $\lambda'_i$ is equal to 0.
As $z\ne y$, we have that $t<1$.
This together with Equation~(\ref{eqn:zViolation}) yields
\[
f(z) > tf(y) + \lambda_1' f(x_1) + \cdots + \lambda_k' f(x_k).
\]
This contradicts inclusion-wise minimality of $x_1,\dots,x_k$.

\end{document}